\documentclass[11pt,USletter]{article}
\usepackage{fullpage}

\usepackage[utf8]{inputenc}
\usepackage{amsthm,amsmath,amssymb}
\usepackage{thmtools}
\usepackage{subcaption,thm-restate}
\usepackage[mathcal,mathscr]{eucal}
\usepackage{epsfig,graphicx,graphics,color}
\usepackage{caption}
\usepackage{bm}
\usepackage{enumerate}
\usepackage[sort,nocompress]{cite}
\usepackage{hyperref}
\hypersetup{
    colorlinks=true,       
    linkcolor=blue,        
    citecolor=red,         
    filecolor=magenta,     
    urlcolor=cyan,         
    linktocpage=true
}

\usepackage{algorithm}
\usepackage{makecell,multirow}
\usepackage{array}
\usepackage{bbm}
\usepackage[noend]{algpseudocode}
\usepackage{varwidth}
\newcounter{algsubstate}


\theoremstyle{plain}
\newtheorem{thm}{Theorem}

\newtheorem{cl}[thm]{Claim}

\newtheorem{qu}[thm]{Question}

\theoremstyle{definition}

\newtheorem{rem}[thm]{Remark}

\def\final{0}  
\def\iflong{\iffalse}
\ifnum\final=0  
\usepackage[dvipsnames]{xcolor}
\newcommand{\kristof}[1]{{\color{red}[{\tiny \textbf{Kristóf:} \bf #1}]\marginpar{\color{red}*}}}
\newcommand{\erika}[1]{{\color{red}[{\tiny \textbf{Erika:} \bf #1}]\marginpar{\color{red}*}}}
\newcommand{\bendre}[1]{{\color{red}[{\tiny \textbf{Endre:} \bf #1}]\marginpar{\color{red}*}}}
\newcommand{\kaz}[1]{{\color{red}[{\tiny \textbf{Kaz:} \bf #1}]\marginpar{\color{red}*}}}
\newcommand{\kavitha}[1]{{\color{red}[{\tiny \textbf{Kavitha:} \bf #1}]\marginpar{\color{red}*}}}
\newcommand{\yusuke}[1]{{\color{red}[{\tiny \textbf{Yusuke:} \bf #1}]\marginpar{\color{red}*}}}
\newcommand{\naoyuki}[1]{{\color{red}[{\tiny \textbf{Naoyuki:} \bf #1}]\marginpar{\color{red}*}}}
\else 
\newcommand{\kristof}[1]{}
\newcommand{\erika}[1]{}
\newcommand{\bendre}[1]{}
\newcommand{\kaz}[1]{}
\newcommand{\kavitha}[1]{}
\newcommand{\yusuke}[1]{}
\newcommand{\naoyuki}[1]{}
\fi

\newcommand{\gse}{EFX$^+_-$}
\newcommand{\gese}{EFX$^0_-$}
\newcommand{\gsee}{EFX$^+_0$}
\newcommand{\gesee}{EFX$^0_0$}

\newcommand\ChangeRT[1]{\noalign{\hrule height #1}}

\newcommand{\PreserveBackslash}[1]{\let\temp=\\#1\let\\=\temp}
\newcolumntype{C}[1]{>{\PreserveBackslash\centering}p{#1}}
\newcolumntype{R}[1]{>{\PreserveBackslash\raggedleft}p{#1}}
\newcolumntype{L}[1]{>{\PreserveBackslash\raggedright}p{#1}}

\makeatletter
\let\@fnsymbol\@arabic
\makeatother

\makeatletter
\def\namedlabel#1#2{\begingroup
    #2%
    \def\@currentlabel{#2}%
    \phantomsection\label{#1}\endgroup
}
\makeatother

\title{Envy-free Relaxations for Goods, Chores, and Mixed Items}

\author{
Kristóf Bérczi\thanks{MTA-ELTE Egerváry Research Group, Department of Operations Research, Eötvös Loránd University, Budapest, Hungary. Email: \texttt{berkri@cs.elte.hu, koverika@cs.elte.hu}.}
\and
Erika R. Bérczi-Kovács\footnotemark[1]
\and
Endre Boros\thanks{MSIS Department and RUTCOR, Rutgers University, New Jersey, USA. Email: \texttt{endre.boros@rutgers.edu}.}
\and
Fekadu Tolessa Gedefa\thanks{Department of Operations Research, Eötvös Loránd University, Budapest, Hungary. Email: \texttt{fekadu\_tolessa@slu.edu.et}.}
\and
Naoyuki Kamiyama\thanks{Institute of Mathematics for Industry, Kyushu University, Fukuoka, Japan, and JST, PRESTO, Kawaguchi, Japan. Email: \texttt{kamiyama@imi.kyushu-u.ac.jp}.}
\and
Telikepalli Kavitha\thanks{School of Technology and Computer Science, Tata Institute of Fundamental Research, Mumbai, India. Email: \texttt{kavitha@tifr.res.in}.}
\and
Yusuke Kobayashi\thanks{Research Institute for Mathematical Sciences (RIMS) Kyoto University, Kyoto, Japan. Email: \texttt{yusuke@kurims.kyoto.ac.jp, makino@kurims.kyoto.ac.jp}.}
\and
Kazuhisa Makino\footnotemark[6]
}

\begin{document}
\maketitle

\begin{abstract}
  In fair division problems, we are given a set $S$ of $m$ items and a set $N$ of $n$ agents with individual preferences, and the goal is to find an allocation of items among agents so that each agent finds the allocation fair. There are several established fairness concepts and envy-freeness is one of the most extensively studied ones. However envy-free allocations do not always exist when items are indivisible and this has motivated relaxations of envy-freeness: envy-freeness up to one item (EF1) and envy-freeness up to any item (EFX) are two well-studied relaxations. We consider the problem of finding EF1 and EFX allocations for utility functions that are not necessarily monotone, and propose four possible extensions of different strength to this setting. 
  
  In particular, we present a polynomial-time algorithm for finding an EF1 allocation for two agents with arbitrary utility functions. An example is given showing that EFX allocations need not exist for two agents with non-monotone, non-additive, identical utility functions. However, when all agents have monotone (not necessarily additive) identical utility functions, we prove that an EFX allocation of chores always exists. As a step toward understanding the general case, we discuss two subclasses of utility functions: Boolean utilities that are $\{0,+1\}$-valued functions, and negative Boolean utilities that are $\{0,-1\}$-valued functions. For the latter, we give a polynomial time algorithm that finds an EFX allocation when the utility functions are identical.

\medskip

\end{abstract}

\section{Introduction} \label{sec:intro}

Fair division of items among competing agents is an important and well-studied problem in Economics and Computer Science. There is a set $S$ of $m$ indivisible items and a set $N$ of $n$ agents with individual preferences, and the goal is to find an allocation of items among agents so that each agent finds the allocation fair. There are several fairness concepts in the literature that use different metrics to measure the degree of equity or fairness. Among them, the possibly most compelling one is envy-freeness (EF). Every agent has a utility value for each subset of items: the $i$-th agent has a utility function $u_i: 2^S \rightarrow \mathbb{R}$. An allocation is envy-free if each agent finds the utility of her bundle at least as much as that of any other agent. Although envy-freeness provides a very natural criterion for the fairness of an allocation, such a solution rarely exists as shown by the following simple example: just consider an instance with two agents and a single item having positive utility for both of them. In order to circumvent the possible non-existence of envy-free allocations, two natural relaxations of envy-freeness have been studied.

In the setting where all utilities are non-negative (items are called {\em goods} in such a case) 
Lipton et al. \cite{lipton2004approximately} and Budish \cite{budish2011combinatorial} introduced the notion of envy-freeness up to one good (EF1) which allows the presence of envies between agents -- however any envy must be removable by deleting some good from the envied agent's bundle.  
Note that no good is really removed from the envied agent's bundle, this is just a thought experiment to measure the amount of envy the envious agent has towards the envied agent.

EF1 may seem too weak as a relaxation of envy-freeness. Consider the simple example given below where there are three goods $a, b, c$ and two agents with {\em additive}\footnote{For \emph{additive} utilities, the utility of a subset $X$ of items is the sum of utilities of items in $X$.} utilities. Both agents value $c$ twice as much as $a$ or $b$. The allocation where agent~1 gets $\{a\}$ and agent~2 gets $\{b,c\}$ is EF1 -- however this allocation does not seem fair towards agent~1.

\begin{center}
\begin{minipage}[b]{0.3\linewidth}
\centering
\begin{eqnarray*}
  \setlength{\arraycolsep}{0.5ex}\setlength{\extrarowheight}{0.25ex}
\begin{array}{@{\hspace{1ex}}c@{\hspace{1ex}}||@{\hspace{1ex}}c@{\hspace{1ex}}|@{\hspace{1ex}}c@{\hspace{1ex}}|@{\hspace{1ex}}c@{\hspace{1ex}}|@{\hspace{1ex}}c@{\hspace{1ex}}}
    \  & a \ & b \ & c \ \\[.5ex] \hline
    Agent~1 \ & 1 \ & 1 \ & 2 \ \\[.5ex] \hline
    Agent~2 \ & 1 \ & 1 \ & 2 \ \\[.5ex] 
\end{array}
\end{eqnarray*}
\end{minipage}
\end{center}

As a less permissive relaxation, Caragiannis et al. \cite{caragiannis2016unreasonable} proposed envy-freeness up to any good (EFX), which requires that an agent's envy towards another bundle can be eliminated by removing {\em any} good from the envied bundle. Observe the subtle difference between the above definitions: when utility functions are  additive, EF1 requires that for any pair of agents $i,j$ agent~$i$'s envy towards agent~$j$ can be eliminated by removing agent~$i$'s most valued good from agent~$j$'s bundle, while EFX requires that the same should hold even after removing agent~$i$'s least valued good from agent $j$'s bundle. In the example given above, the allocation where one agent gets $\{a,b\}$ and the other agent gets $\{c\}$ is the only EFX allocation. Thus EFX, though strictly weaker than EF, is strictly stronger than EF1. As remarked in \cite{CGH19}: ``{\em Arguably, EFX is the best fairness analog of envy-freeness of indivisible items}''.

Most works on fair division focused on allocation of `goods', i.e., utilities are always non-negative and monotone: the utility of any set is at least as much as each of its subsets. However, in practice, it might easily happen that an item is a {\em chore}, i.e., this is a task or responsibility that all agents find cumbersome, so the utility of a chore is zero or negative -- hence utilities are no longer monotone. We can also have {\em mixed} items: so an item has positive utility for some agent and negative for another agent. Furthermore, the utility functions are not necessarily additive, hence the utility of a subset of items might be completely independent of the utilities of the individual items within. 

Aziz et al. \cite{aziz2019fair} initiated the study of fair allocations of indivisible goods and chores and extended the above fairness concepts to deal with chores, mixed items and more general utility functions. Our work continues to discover this direction by showing the existence or non-existence of fair solutions in various settings. We also focus on the setting where all items are chores: so all utilities are non-positive here. As in the case of goods, envy-free allocations of chores need not exist. Consider the simple example with two agents and a single chore that neither agent wants to do: one of the agents has to be assigned this chore and she envies the other agent who is assigned no chore. 

We formally define EF1 and EFX allocations in the generalized setting of goods and chores in Section~\ref{sec:prelim}. Roughly speaking, an allocation $\pi = \langle \pi(1),\ldots,\pi(n)\rangle$ is EF1 if for any pair of agents $i,j$: either (1)~$i$ does not envy $j$'s bundle, or (2)~there is some item $s \in \pi(i) \cup \pi(j)$ such that $i$ likes $\pi(i)-s$ at least as much as $\pi(j) - s$. Similarly, $\pi$ is EFX if for any pair of agents $i,j$: (1)~$i$ does not envy $j$'s bundle, or (2)~$i$ likes $\pi(i)$ at least as much as $\pi(j) - s$ for any `good' $s \in \pi(j)$ and $i$ likes $\pi(i) - s$ at least as much as $\pi(j)$ for any `chore' $s \in \pi(i)$. Note that in the general case, whether an item $s$ is a good or chore in $X \subseteq S$ is determined by the marginal utility that $s$ brings to the set $X - s$. We refer to Section~\ref{sec:prelim} for details.

\subsection{Previous Work}

We only discuss works that are closely related to our results, for a detailed survey of the fair division literature we refer the reader to \cite{brams1996fair,brandt2016handbook}. 

\paragraph{EF1 allocations.} Envy-freeness and its relaxations were mainly considered for monotone additive utility functions. The idea of envy-freeness up to one good (EF1) implicitly appeared in the paper of Lipton et al.  \cite{lipton2004approximately}, and then was explicitly introduced and analyzed by Budish \cite{budish2011combinatorial}.

When allocating divisible goods, the maximum Nash welfare solution selects an allocation that maximizes the product of utilities and it is known to have strong fairness guarantees; moreover it also satisfies Pareto optimality.\footnote{An allocation is \emph{Pareto optimal} if there is no other allocation which makes some agent strictly better off while making no agent worse off.} Caragiannis et al. \cite{caragiannis2016unreasonable} showed that the maximum Nash welfare solution is fair in the indivisible setting as well -- in the sense that it forms an EF1 allocation. In \cite{barman2018finding}, Barman et al. developed a pseudo-polynomial time algorithm for finding EF1 allocations of goods that are also Pareto efficient, and showed that there always exists an allocation that is EF1 and fractionally Pareto efficient.

Much less is known when the utility functions are non-additive or non-monotone. Caragiannis et al. \cite{caragiannis2012efficiency} considered allocation problems in which a set of divisible or indivisible goods or chores has to be allocated among several agents. Aziz et al. \cite{aziz2016optimal} focused on additive cardinal utility functions, and presented computational hardness results as well as  polynomial-time algorithms for testing Pareto optimality for subclasses of utility functions. Barman and Murthy \cite{barman2017approximation} considered additive and submodular utilities for goods, and provided fairness guarantees in terms of the maximin share, i.e., the maximum value that an agent can ensure for herself by partitioning the goods into $n$ bundles, and receiving a minimum valued bundle. Benabbou et al. \cite{benabbou2020finding} considered the allocation of indivisible goods to agents that have monotone, submodular, non-additive utility functions. They showed that utilitarian socially optimal (hence Pareto optimal), leximin, and maximum Nash welfare allocations are all EF1 if in addition the utility functions have binary marginal gain.

Bogomolnaia et al. \cite{bogomolnaia2019dividing} compared the performance of the Competitive and the Egalitarian division rules for problems involving divisible goods and chores. Recently, Aziz et al. \cite{aziz2019fair} proposed a similar scenario in the indivisible setting, and gave a generalization of the decentralized round robin algorithm that finds an EF1 allocation when the utilities are additive. They also presented a different polynomial-time algorithm that returns an EF1 allocation even when the utility functions are arbitrary. However, there is a flaw in the proof of the latter result, thus the existence of an EF1 allocation in the non-monotone, non-additive setting is still open even for identical utility functions. 

\begin{table}[htbp] 
\small
    \centering
    {\setlength{\extrarowheight}{4pt}
    \begin{tabular}{!{\vrule width 1.5pt}C{0.08\linewidth}
                    !{\vrule width 1pt}C{0.18\linewidth}
                    |C{0.16\linewidth}
                    !{\vrule width 1pt}C{0.18\linewidth}
                    |C{0.26\linewidth}
                    !{\vrule width 1.5pt}}
        \ChangeRT{1.5pt}
        \multirow{2}{*}{Utilities} & \multicolumn{2}{c!{\vrule width 0.8pt}}{Additive} & \multicolumn{2}{c!{\vrule width 1.5pt}}{Non-additive}\\ [3pt] 
        \cline{2-5}
         & Monotone & Non-monotone & Monotone & Non-monotone \\[3pt]
        \ChangeRT{1pt}
        Identical     & \multirow{2}{*}[-1em]{Exists$^{\mbox{\tiny\cite{lipton2004approximately}}}$} & \multirow{2}{*}[-1em]{Exists$^{\mbox{\tiny\cite{aziz2019fair}}}$} & \multirow{2}{*}[-1em]{Exists$^{\mbox{\tiny\cite{lipton2004approximately}}}$} & \makecell{Open\\ \textbf{\bm{$n=2$} (Thm.~\ref{thm:ef1})}\\\textbf{Boolean (Thm. \ref{thm:boolean})}\\\textbf{Neg. Boolean} \textbf{(Thm. \ref{thm:negboolean})}} \\[3pt]
        \cline{1-1}\cline{5-5}
        Non-identical &  &  &  & \makecell{Open\\ \textbf{\bm{$n=2$} (Thm.~\ref{thm:ef1})}\\\textbf{Boolean (Thm. \ref{thm:boolean})}}  \\[5pt]
        \ChangeRT{1.5pt}
    \end{tabular}}
\caption{Landscape of results on EF1 allocations, where bold letters represents the results obtained in the present paper. A utility function $u$ is called \emph{Boolean} if $u(X)\in\{0,1\}$ for every $X\subseteq S$, and \emph{negative Boolean} if $-u$ is Boolean.}
\label{tab:ef1res}
\end{table}

\paragraph{EFX allocations.} 
As a strictly stronger fairness criterion than EF1, \cite{caragiannis2016unreasonable} introduced envy-freeness up to any good (EFX) and provided a connection to pairwise maximum share guarantee for additive utilities over indivisible goods. Plaut and Roughgarden \cite{plaut2018almost} investigated a stronger variant of EFX, referred to as EFX$_0$ in \cite{kyropoulou2019almost}, for monotone utilities in the setting of goods. 
They proved that an EFX$_0$ solution always exists when the utility functions are identical. This result yields a protocol that produces an EFX$_0$ allocation of goods for two players with general and possibly distinct utilities. Barman and Murthy \cite{barman2017approximation} extended their result by proving the existence of EFX$_0$ allocations for additive utility functions when all agents have the same preference toward goods. In a recent paper, Chaudhury et al. \cite{chaudhury2020efx} verified the the existence of an EFX$_0$ solution for three agents with monotone, additive, non-identical utilities. For monotone submodular utilities with binary marginal gains, Benabbou et al.\cite{benabbou2020finding} noted that EF1 implies EFX, thus their results guarantee the existence of EFX allocations (that are Pareto optimal as well) for any such instance. However, they showed that even an EF1 and utilitarian optimal allocation may violate the EFX$_0$ condition. Amanatidis et al. \cite{amanatidis2020maximum} studied $2$-value instances, that is, when there are at most two possible values for the goods, and proved that any allocation that maximizes the Nash welfare is EFX$_0$. They proposed an algorithm called {\it Match \& Freeze} for finding an EFX$_0$ solution that is based on repeatedly computing maximum matchings and freezing certain agents. Furthermore, they gave an algorithm for instances where the utilities of each agent takes values in an interval such that the ratio between the maximum and the minimum value is at most 2.

In general, most of the results on EFX$_0$ allocations of goods cannot be straightforwardly translated into analogous results on allocations of chores. Therefore only a few results are known for the setting of mixed items when both goods and chores are present. A utility function is tertiary if $u(s)\in\{-1,0,+1\}$ for every $s\in S$. It is not difficult to see that for tertiary utilities EF1 implies EFX, thus the \emph{Double Round Robin Algorithm} of \cite{aziz2019fair} provides such a solution. For the same setting, a different algorithm that provides a solution which is Pareto optimal as well was given by Alexandrov and Walsh \cite{aleksandrov2019greedy}. They also showed that an EFX allocation can be found in polynomial time for identical additive utilities. However, these results seems to be difficult to extend to the EFX$_0$ case. Very recently, Chen and Liu \cite{chen2020fairness} analyzed the fairness of leximin solutions in allocation of indivisible chores. Their model contained itemwise monotone utility functions, meaning that each agent $i$ can partition $S$ into sets $G_i$ and $C_i$ such that for any $X\subseteq S$ we have $v_i(X)\leq v_i(X+s)$ if $s\in G_i$ and $v_i(X)\geq v_i(X+s)$ for $s\in C_i$. They verified that a leximin solution is EFX for combinations of goods and chores for agents with identical itemwise monotone utilities. However, their notion of EFX solutions is different from ours and so none of our results are implied by their paper.

Gourv\`es et al. \cite{gourves2014near} and Freeman et al. \cite{freeman2019equitable} studied \emph{equitable} (EQ) allocations\footnote{An allocation $\pi$ is \emph{equitable} if $v_i(\pi(i)) = v_j(\pi(j))$ for all agents $i,j$.} and their relaxations, \emph{equitable up to one good} (EQ1) and \emph{equitable up to any good} (EQX). Freeman et al. \cite{FSVX20} introduced analogous notions of (approximate) equitability of chores. By combining the results of \cite{gourves2014near} and \cite{FSVX20}, an algorithm for finding an EQX allocation of mixed items for additive utilities is at hand. When agents have identical valuations, an allocations satisfies EFX if and only it satisfies EQX, thus this immediately implies the existence of an EFX allocation of mixed items for additive identical valuations. By the protocol in \cite{plaut2018almost}, this result guarantees EFX allocations of mixed items for two agents with additive (not necessarily identical) utilities.

Despite all the efforts made, the existence of an EFX/EFX$_0$ allocation remains an intriguing open question for at least four agents with distinct, monotone, additive utility functions.

\begin{table}[htbp]
    \small
    \centering
    {\setlength{\extrarowheight}{4pt}
    \begin{tabular}{!{\vrule width 1.5pt}C{0.08\linewidth}
                    !{\vrule width 1pt}C{0.18\linewidth}
                    |C{0.16\linewidth}
                    !{\vrule width 1pt}C{0.18\linewidth}
                    |C{0.26\linewidth}
                    !{\vrule width 1.5pt}}
        \ChangeRT{1.5pt}
        \multirow{2}{*}{Utilities} & \multicolumn{2}{c!{\vrule width 0.8pt}}{Additive} & \multicolumn{2}{c!{\vrule width 1.5pt}}{Non-additive}\\ [3pt] 
        \cline{2-5}
         & Monotone & Non-monotone & Monotone & Non-monotone \\[3pt]
        \ChangeRT{1pt}
        Identical     & \makecell{Exists$^{\mbox{\tiny\cite{plaut2018almost}}}$ (goods)\\Exists$^{\mbox{\tiny\cite{gourves2014near}}}$ (chores)} & Exists$^{\mbox{\tiny\cite{gourves2014near,FSVX20}}}$ & \makecell{Exists$^{\mbox{\tiny\cite{plaut2018almost}}}$ (goods)\\\textbf{Exists (Thm. \ref{thm:chores-identical-val})}\\\textbf{(chores)}} & \textbf{\makecell{Not exists (Thm. \ref{thm:nm_na_i})\\Boolean (Thm. \ref{thm:boolean})\\Neg. Boolean (Thm. \ref{thm:negboolean})}}\\ [3pt]
        \cline{1-5}
        Non-identical & \makecell{Open\\$n=3^{\mbox{\tiny\cite{chaudhury2020efx}}}$ \\ Tertiary$^{\mbox{\tiny\cite{aziz2019fair}}}$} & \makecell{Open\\$n=2^{\mbox{\tiny\cite{gourves2014near,FSVX20}}}$\\ Tertiary$^{\mbox{\tiny\cite{aziz2019fair}}}$} & \makecell{Open\\$n=2^{\mbox{\tiny\cite{plaut2018almost}}}$ (goods)\\$(0,1)$-SUB$^{\mbox{\tiny\cite{benabbou2020finding}}}$\\\textbf{\bm{$n=2$} (Cor. \ref{cor:chores-two-agents})}\\\textbf{(chores)}} & \textbf{\makecell{Not exists (Thm. \ref{thm:nm_na_i})\\Boolean (Thm. \ref{thm:boolean})}} \\ [5pt]
        \ChangeRT{1.5pt}
    \end{tabular}}
     \caption{Landscape of results on different variants of EFX allocations, where bold letters represents the results obtained in the present paper. A utility function $u$ is called \emph{Boolean} if $u(X)\in\{0,1\}$ for every $X\subseteq S$, and \emph{negative Boolean} if $-u$ is Boolean.}
\label{tab:efxres}
\end{table}

\subsection{Our Results}

We study the existence or non-existence of EF1 and EFX allocations for non-monotone utility functions. Aziz et al. \cite{aziz2019fair} presented a polynomial-time algorithm that returns an EF1 allocation for arbitrary utilities, but there is a flaw in the proof (an example is described in Section~\ref{sec:ef1}). Our first main result is a polynomial-time algorithm that finds an EF1 allocation for two agents with arbitrary (i.e. non-monotone, non-additive, non-identical) utility functions.

\begin{restatable}{thm}{thmefo}
\label{thm:ef1}
There always exists an EF1 allocation for two agents with arbitrary (not necessarily monotone or additive) utility functions. Such an allocation can be computed in polynomial time.
\end{restatable}

Prior work on EFX/EFX$_0$ allocations has mainly concentrated on the cases when all items are either goods or chores and accordingly, the notion of EFX solutions was not defined for the mixed setting. It turns out that finding the right generalization for non-additive, non-monotone utility functions is not straightforward. In Section~\ref{sec:prelim}, we propose four possible extensions possessing different characteristics, called \gesee, \gese, \gsee, and \gse\footnote{We remark that \gse and \gese, when restricted to monotone, additive utilities, correspond to the usual notion of EFX and EFX$_0$ solutions, respectively.}; here \gesee\ implies both \gese\ and \gsee, and both \gese\ and \gsee\ imply \gse. The definitions differ in whether items with a marginal value of $0$ are taken into account or not, both on the side of the envious and the envied agents. Unfortunately, Theorem~\ref{thm:ef1} does not generalize even to \gse\ allocations. This is so even when all the utility functions are identical.

\begin{restatable}{thm}{thmnmnai}
\label{thm:nm_na_i}
There need not exist an \gse\ allocation for two agents with non-monotone, non-additive, identical utility functions.
\end{restatable}

In the setting of goods (all utilities are non-negative), it is known~\cite{plaut2018almost} that EFX$_0$ allocations always exist for monotone identical utility functions. We show that in the setting of chores as well, \gsee\ allocations always exist for monotone identical utility functions: note that a monotone utility function $u$ in the setting of chores is monotone non-increasing, i.e., $X \subseteq Y$ implies $u(X) \ge u(Y)$.

Equitable allocations (EQ) and relaxations (EQ1 and EQX) of chores were studied in \cite{FSVX20}. It was shown there that when all utilities are additive and identical then an \gse\ allocation always exists. We strengthen this by showing that an \gsee\ allocation exists for all monotone identical utility functions.

\begin{restatable}{thm}{thmchores}
\label{thm:chores-identical-val}
When all agents have monotone (not necessarily additive) identical utility functions, an \gsee\ allocation of chores always exists.
\end{restatable}

By using the cut-and-choose-based protocol of Plaut and Roughgarden \cite{plaut2018almost}, the above result guarantees an \gsee\ allocation of chores for two agents with monotone non-identical utility functions. 

\begin{restatable}{cor}{corchorestwoagents}
\label{cor:chores-two-agents}
There always exists an \gsee\ allocation of chores for two agents with monotone (not necessarily additive) utility functions. Such an allocation can be computed in polynomial time.
\end{restatable}

As a further step towards understanding the general case, we consider special subclasses of utility functions. \emph{Boolean utilities} are non-monotone, non-additive $\{0,1\}$-valued functions. A function $u$ is \emph{negative Boolean} if $-u$ is Boolean. 

\begin{restatable}{thm}{thmboolean}
\label{thm:boolean}
When all agents have Boolean utility functions, an \gese\ allocation always exists.
\end{restatable}

\begin{restatable}{thm}{thmnegboolean}
\label{thm:negboolean}
When all agents have identical negative Boolean utility functions, an \gsee\ allocation always exists.  Such an allocation can be computed in polynomial time.
\end{restatable}

Although these two statements show a lot of similarities, their proofs will be significantly different: while the Boolean case is an easy observation, the negative Boolean case requires a somewhat tricky approach that only works for the case of identical utility functions.

\subsection{Techniques: Old and New}
\label{sec:techniques}

\paragraph{EF1 allocations.}
In the setting of goods and monotone utilities, an EF1 allocation always exists~\cite{lipton2004approximately}. The algorithm in  \cite{lipton2004approximately} assigns goods to bundles one-by-one, while maintaining the property that the partial allocation constructed so far is EF1. At a general step of the algorithm, if no unenvied agent exists, then there exists a cycle of envy among agents, and bundles are shifted along such cycles until no cycle of envy remains. Otherwise the next good is added to an unenvied agent. The key observation is that although such a step might result in new envies, those can be eliminated by removing the good she just received, so the resulting partial allocation remains EF1.

In order to find an EF1 allocation with arbitrary utilities, the Generalized Envy Graph Algorithm proposed in \cite{aziz2019fair} generalizes the above algorithm. The items are allocated to agents one-by-one, always maintaining the property that the partial allocation constructed so far is EF1. In a general step of the algorithm, the next item is added to a source vertex of the subgraph of the envy graph spanned by the set of agents who have a non-negative marginal utility for the given item. If no such agent exists, then the item is added to a sink vertex of the envy graph. After the addition of the item, directed cycles of the modified envy graph are eliminated by shifting the bundles along them. However the allocation provided by the algorithm is not necessarily EF1, even when the utility functions are identical -- we include a simple example illustrating this in Section~\ref{sec:ef1}.

In Section~\ref{sec:ef1}, we prove that an EF1 allocation always exists for $n=2$ with arbitrary utilities. More specifically, we show that for any ordering $s_1,\ldots,s_m$ of all the items in the entire set $S$, there is an EF1 allocation $\pi = \langle \pi(1),\pi(2)\rangle$ such that $\{\pi(1),\pi(2)\} = \{\{s_1,\ldots,s_t\}, \{s_{t+1},\ldots,s_m\}\}$ for some $0 \le t \le m$. This is shown by a combinatorial argument on prefixes and suffixes of the ordered set $\langle s_1,\ldots,s_m\rangle$. We show an example in Section~\ref{sec:counter-example} that an EFX allocation need not exist for 2 agents with non-monotone, non-additive utilities - even when the utility functions are identical.

\paragraph{EFX allocations.}
As mentioned earlier, Plaut and Roughgarden \cite{plaut2018almost} showed the existence of EFX$_0$ allocations for monotone (non-decreasing) utilities in the setting of goods. They construct such an allocation using the so-called leximin solution that selects the allocation which maximizes the minimum utility, and subject to this, maximizes the second minimum utility, and so on. Our proof of existence of \gsee\ allocations for monotone (non-increasing) utilities in the setting of chores is different from this and is similar to the method used in \cite{CKMS20} where another proof of the above result from \cite{plaut2018almost} was given - this method also leads to a pseudo-polynomial time algorithm to find an \gsee\ allocation. 

Our construction resembles the one in \cite{lipton2004approximately} to show the existence of EF1 allocations. Interestingly, the construction in \cite{lipton2004approximately} never {\em breaks up} any bundle obtained in a partial allocation whereas the construction in \cite{CKMS20} and ours may need to -- this is because we need to maintain an EFX allocation which is more demanding than an EF1 allocation. Thus items allocated to agents in earlier rounds may go back to the pool of unallocated items in some later round. However we show a potential function (same as the one in \cite{CKMS20}) that improves as our algorithm progresses -- thus our algorithm always converges. This result is given in Section~\ref{sec:chores-efx}.

The solutions for Boolean and negative Boolean utilities are based on different approaches. While the Boolean case is an easy observation but is non-algorithmic, the negative Boolean case is solved via a non-trivial improvement step. 
Our results for these special cases are given in Sections~\ref{sec:boolean} and \ref{sec:negboolean}, respectively. We hope that the tools used for these special cases will prove useful in other situations as well. Finally, the main results and the most important open problems are summarized in Section~\ref{sec:conc}.

\subsection{Other Related Work}

Let us further mention several results that are closely related to fair divisions of indivisible items. Aziz et al. \cite{aziz2015fair} considered an assignment problem in which agents have ordinal preferences over objects and these objects are allocated to the agents in a fair manner. They introduced several proportionality and envy-freeness concepts for discrete assignments, and gave polynomial-time algorithms to check whether a fair assignment exists for several of these fairness notions. In \cite{brams2017maximin}, Brams et al. discussed a problem setting with two agents who have strict rankings over an even number of indivisible items. They proposed algorithms to find balanced allocations of these items that maximize the minimum rank of the items that the agents receive, and are envy-free and Pareto optimal, if such allocations exist. In \cite{plaut2019communication}, Plaut and Roughgarden provided a throughout description of the communication complexity of computing a fair allocation with indivisible goods for every combination of fairness notion, utility function class, and number of players.

Another line of research considered group envy-freeness instead of pairs of agents. Berliant et al. \cite{berliant1992fair} generalized envy-freeness for equal-sized groups of agents. Conitzer et al. \cite{conitzer2019group} introduced the concept of group fairness, which implies most existing notions of individual fairness. They further proposed two relaxations similar to EF1, and showed that certain local optima of the Nash welfare function satisfy both relaxations and can be computed in pseudo-polynomial time by local search. However, \cite{conitzer2019group} assumed only goods, that is, items for which agents have positive utility. While the notion of group fairness and group envy-freeness can be extended to the case when chores are also present, the same does not hold for the relaxations. Aziz et al. \cite{aziz2018fair} proposed fairness concepts that are suitable to handle the case of goods and for chores as well. In the same spirit, Aziz and Rey \cite{aziz2019almost} defined several variants and relaxations of group fairness and group envy-freeness when both goods and chores are present. 

\section{Preliminaries}
\label{sec:prelim}

\subsection{Basic Notation}

Throughout the paper, $N$ denotes a set of $n$ \emph{agents} and $S$ denotes a set of $m$ indivisible \emph{items}. For simplicity, we will denote a subset of items by simply enumerating its elements without separating them, e.g. $124$ stands for the set $\{1,2,4\}$.

For each agent $i\in N$, a \emph{utility function} $u_i:2^S\rightarrow\mathbb{R}$ is given that represents agent $i$'s preferences over the subsets of items. We always assume that the empty set has value $0$, that is, $u_i(\emptyset)=0$ for $i\in N$. We say that the utility functions are \emph{identical} if $u_i(X)=u_j(X)$ for every $X\subseteq S$ and $i,j\in N$, and \emph{non-identical} if this condition does not necessarily hold. A utility function $u$ is called \emph{additive} if $u(X)=\sum_{s\in X} u(s)$ for every subset $X\subseteq S$, and it is \emph{monotone non-decreasing} if $u(X)\leq u(Y)$ whenever $X\subseteq Y$; similarly, it is \emph{monotone non-increasing} if $u(X)\geq u(Y)$ whenever $X\subseteq Y$. For simplicity, we will use \emph{non-additive} and \emph{non-monotone} as a shorthand for \emph{not necessarily additive} and \emph{not necessarily monotone}, respectively. 

The \emph{marginal utility} of an item $s\in S$ towards a subset $X\subseteq S$ is denoted by $u(s|X)$ and is defined as $u(s|X)=u(s+X)-u(X)$. For a set $X\subseteq S$ and agent $i$, we denote by $S^+_i(X)=\{s\in X\mid u_i(s|X-s)>0\}$, $S^-_i(X)=\{s\in X\mid u_i(s|X-s)<0\}$, and $S^0_i(X)=\{s\in X\mid u_i(s|X-s)=0\}$ the sets of items in $X$ whose deletion decreases, increases, and does not change the utility of agent $i$ on set $X$, respectively. We omit the index $i$ when the utility functions are identical.

An \emph{allocation} of $S$ is a function $\pi:N\rightarrow 2^S$ assigning to each agent a (possibly empty) subset of items that altogether give a partition of $S$, that is, $\pi(i)\cap\pi(j)=\emptyset$ for distinct $i,j\in N$ and $\bigcup_{i\in N}\pi(i)=S$. We will refer to the set $\pi(i)$ as the \emph{bundle of agent $i$}. Agent $i$ \emph{envies} agent $j$ if $u_i(\pi(i))<u_i(\pi(j))$. To any allocation $\pi$, we associate a directed graph $G_\pi=(N,E)$ called the \emph{envy graph}, where there is a directed edge from $i$ to $j$ if agent $i$ envies agent $j$ for the given allocation.

\subsection{Fairness in the Non-monotone, Non-additive Setting}

There are several ways to characterize fairness, probably the most natural one being \emph{envy-freeness}, that requires that no agent envies another agent: 
\begin{description}
\item[(\namedlabel{ef}{EF})] 
For any $i,j\in N$ inequality $u_i(\pi(i))\geq u_i(\pi(j))$ holds. 
\end{description}

Budish \cite{budish2011combinatorial} provided a relaxation of envy-freeness by introducing the concept of \emph{envy-freeness up to one good} in the context of monotone allocations. Aziz et al. \cite{aziz2019fair} extended the definition to the non-monotone case by requiring that an agent's envy can be eliminated by removing some item either from her own bundle or the envied one:
\begin{description}
\item[(\namedlabel{ef1}{EF1})]
For any $i,j\in N$ at least one of the following holds:
\begin{enumerate}[(i)]
    \item $u_i(\pi(i))\geq u_i(\pi(j))$ \label{eq:ef1i}
    \item  $u_i(\pi(i) - s)\geq u_i(\pi(j)-s)$ for some $s\in \pi(i) \cup \pi(j)$. 
\end{enumerate}
\end{description}
Note that EF1 is strictly weaker than EF.

As a less permissive relaxation of envy-freeness, Caragiannis et al. \cite{caragiannis2016unreasonable} introduced the notion of \emph{envy-freeness up to any good} in the context of goods and monotone utilities. According to their definition, an allocation is EFX if for any pair $i,j$ of agents, agent~$i$ may envy agent~$j$, however this envy would vanish upon removing any good from $j$'s bundle, i.e., $u_i(\pi(i)) \ge u_i(\pi(j) - s)$ for all $s \in S^+_i(\pi(j))$. Plaut and Roughgarden \cite{plaut2018almost} introduced a stronger variant called EFX$_0$, where the envy should vanish upon removing any good from $S^+_i(\pi(j))\cup S^0_i(\pi(j))$.
For non-monotone, additive utility functions with goods and chores Aziz et al. \cite{aziz2019fair} proposed to call an allocation EFX if for any pair $i,j$ of agents, $u_i(\pi(i)-s)\geq u_i(\pi(j)-s)$ for any $s\in S^+_i(\pi(j))\cup S^-_i(\pi(i))$. Note that the latter definition does not require anything to hold for items $s\in S$ with $0$ marginal utility value.

Extending the definition to non-monotone, non-additive utilities is not immediate. We introduce four variants, in descending order of strength.

\begin{description}
\item[(\namedlabel{gesee}{\gesee})]
For any $i,j\in N$ at least one of the following holds:
\begin{enumerate}[(i)]
    \item $u_i(\pi(i))\geq u_i(\pi(j))$ \label{eq:geseei}
    \item $u_i(\pi(i))\geq u_i(\pi(j)-s)$ for every $s\in S^+_i(\pi(j))\cup S^0_i(\pi(j))$,\\
          $u_i(\pi(i)-s)\geq u_i(\pi(j))$ for every $s\in S^-_i(\pi(i))\cup S^0_i(\pi(i))$, and\\
          $S^+_i(\pi(j))\cup S^0_i(\pi(j))\cup S^-_i(\pi(i))\cup S^0_i(\pi(i))\neq\emptyset$.
    \label{eq:geseeii}
\end{enumerate}
\end{description}

It is not difficult to see that \gesee is strictly weaker than EF, but it is strictly stronger than EF1. However, \gesee\ is too much to ask for: with such a definition, the simple example with two agents and two goods with identical additive utilities $u(1) = 1$ and $u(2) = 0$ has no \gesee\ allocation. Thus we introduce two, slightly weaker variants.

\begin{description}
\item[(\namedlabel{gese}{\gese})]
For any $i,j\in N$ at least one of the following holds:
\begin{enumerate}[(i)]
    \item $u_i(\pi(i))\geq u_i(\pi(j))$ \label{eq:gesei}
    \item $u_i(\pi(i))\geq u_i(\pi(j)-s)$ for every $s\in S^+_i(\pi(j))\cup S^0_i(\pi(j))$,\\
          $u_i(\pi(i)-s)\geq u_i(\pi(j))$ for every $s\in S^-_i(\pi(i))$, and\\
          $S^+_i(\pi(j))\cup S^0_i(\pi(j))\cup S^-_i(\pi(i))\neq\emptyset$.
    \label{eq:geseii}
\end{enumerate}
\end{description}

\begin{description}
\item[(\namedlabel{gsee}{\gsee})]
For any $i,j\in N$ at least one of the following holds:
\begin{enumerate}[(i)]
    \item $u_i(\pi(i))\geq u_i(\pi(j))$ \label{eq:gseei}
    \item $u_i(\pi(i))\geq u_i(\pi(j)-s)$ for every $s\in S^+_i(\pi(j))$,\\
          $u_i(\pi(i)-s)\geq u_i(\pi(j))$ for every $s\in S^-_i(\pi(i))\cup S^0_i(\pi(i))$, and\\
          $S^+_i(\pi(j))\cup S^-_i(\pi(i))\cup S^0_i(\pi(i))\neq\emptyset$.
    \label{eq:gseeii}
\end{enumerate}
\end{description}
\gese and \gsee are symmetric, therefore these definitions represent in a certain sense dual concepts. Nevertheless, we will see that results for one of them do not automatically carry over to the other. Note that \gesee implies both \gese and \gsee. 

Finally, let us introduce a further weaker condition that is the easiest to work with.

\begin{description}
\item[(\namedlabel{gse}{\gse})]
For any $i,j\in N$ at least one of the following holds:
\begin{enumerate}[(i)]
    \item $u_i(\pi(i))\geq u_i(\pi(j))$ \label{eq:gsei}
    \item $u_i(\pi(i))\geq u_i(\pi(j)-s)$ for every $s\in S^+_i(\pi(j))$,\\
          $u_i(\pi(i)-s)\geq u_i(\pi(j))$ for every $s\in S^-_i(\pi(i))$, and\\
          $S^+_i(\pi(j))\cup S^-_i(\pi(i))\neq\emptyset$.
    \label{eq:gseii}
\end{enumerate}
\end{description}
Clearly, \gse is implied by both \gese and \gsee. For additive, monotone utility functions, EFX introduced in \cite{caragiannis2016unreasonable} and EFX$_0$ introduced  in \cite{plaut2018almost} are identical to \gse and \gese, respectively. We further note that if $s\in S_i^0(X)$ then $u_i(X-s)=u_i(X)$. Therefore agent $i$ cannot envy agent $j$ in an \gesee\ or \gese\ allocation if $S^0_i(\pi(j))\neq\emptyset$, and similarly, agent $i$ cannot envy agent $j$ in an \gesee\ or \gsee\ allocation if $S^0_i(\pi(i))\neq\emptyset$. In this sense, parts $(i)$ and $(ii)$ of definitions \ref{gesee}, \ref{gese}, and \ref{gsee} describe non-disjoint situations. 

In any of the above cases, if $u_i(\pi(i))<u_i(\pi(j))$ for some allocation $\pi$, then we refer to this envy as an \emph{EFX envy} if $i$ and $j$ satisfy the appropriate version of the second condition in  \gesee -- \gse, otherwise it is called a \emph{non-EFX envy}.

\subsection{Relation to EFX and \texorpdfstring{EFX$_0$}{EFX0}}

Although the notion of \gesee\ allocations seems to be reasonable, such an allocation does not necessarily exist even in very simple examples. \gese\ allocations provide a natural extension of the EFX$_0$ property, and \gsee\ serves as a symmetric counterpart. Finally, \gse\ provides a generalization of EFX allocations. 

Let us explain why the non-emptiness condition in the $(ii)$ part of the definitions is necessary. Consider an example with two agents and two items with identical utilities, namely $u(\emptyset)=0$, $u(1)=u(2)=2$ and $u(12)=1$. Then for the allocation $\pi(1)=\emptyset$ and $\pi(2)=12$, agent 1 envies agent 2. Then all of the sets $S^-_1(\pi(1))$, $S^0_1(\pi(1))$, $S^0_1(\pi(2))$, and $S^+_1(\pi(2))$ are empty, thus condition $(ii)$ without the non-emptiness assumption is a tautology in all cases. Nevertheless, the allocation does not seem to be fair towards agent $1$ in any way.

The peculiar conditions for the non-emptiness of the corresponding sets have not arisen in previous works for two reasons. On the one hand, earlier results mainly focused on monotone (non decreasing) utilities. But most importantly, they mostly considered additive utilities, and in such cases the non-emptiness conditions are redundant. Indeed, if $u_i(\pi(i))<u_i(\pi(j))$ and $u_i$ is additive, then at least one of the sets $S^-_i(\pi(i))$ and $S^+_i(\pi(j))$ is non-empty.

\section{Non-monotone and Non-additive Utilities}
\label{sec:nmna}

In this section we consider arbitrary utilities. Somewhat surprisingly, an EF1 allocation always exists for two agents and can be found in polynomial time even in this general setting; our algorithm is presented in Section~\ref{sec:ef1}. The existence of an EFX allocation with arbitrary but identical utilities is discussed in Section~\ref{sec:counter-example}.

\subsection{\textbf{EF1} allocations}
\label{sec:ef1}

We now show a simple example that the allocation provided by the Generalized Envy Graph Algorithm proposed in \cite{aziz2019fair} (see Section~\ref{sec:techniques} for an outline of this algorithm) is not necessarily EF1, even if the utility functions are identical. Consider an instance with two agents and three items with utilities defined as follows: $u(\emptyset)=0$, $u(1)=1$, $u(2)=-1$, $u(3)=-1$, $u(12)=u(13)=u(23)=1$, and $u(123)=1$. It is not difficult to check that if the items arrive in order $1,2,3$, then this algorithm gives $\pi(1)=123$ and $\pi(2)=\emptyset$ as a solution. This solution is not EF1 as the deletion of any of the items from the bundle of the first agent results in a bundle with strictly positive utility value. Hence Theorem 2 in \cite{aziz2019fair} is not correct in its present form. 

That is, the existence of an EF1 allocation for non-monotone, non-additive utility functions is still open even for identical utilities. Our first result shows that such an allocation exists for two agents in the most general setting. 

\thmefo*
\begin{proof}
For an arbitrary ordering $1,\dots,m$ of the items, let $F_i=\{1\dots i\}$ and $L_i=\{i+1\dots m\}$ denote the sets of the first $i$ and last $m-i$ items, respectively, for $i=0,\dots,m$, where $F_0=L_m=\emptyset$. 

Assume that $u_i(F_j) = u_i(L_j)$ for some $i\in\{1,2\}$ and $j\in\{0,\dots,m\}$. If $u_{3-i}(F_j) \leq u_{3-i}(L_j)$, then let $\pi(i)=F_j$ and $\pi(3-i)=L_j$, else let $\pi(i)=L_j$ and $\pi(3-i)=F_j$. By the assumption $u_i(F_j) = u_i(L_j)$, $\pi$ is an envy-free allocation. Similarly, if there exists an index $j\in\{0,\dots,m\}$ such that $u_i(F_j)\geq u_i(L_j)$ and $u_{3-i}(F_j)\leq u_{3-i}(L_j)$ for some $i\in\{1,2\}$, then $\pi(i)=F_j$, $\pi(3-i)=L_j$ is an envy-free allocation again. 

From now on we assume that neither of the above two cases holds. Thus $u_1(S) \ne 0$. We distinguish two cases.

\medskip

\noindent \textbf{Case 1.} $u_1(S)>0$ where $S$ is the entire set of $m$ items.

As $u_1(F_0)=u_1(L_m)=u_1(\emptyset)=0$ and $u_1(F_m)=u_1(L_0)=u_1(S)>0$, there exists an index $j \in \{0,\ldots,m-1\}$ such that $u_1(F_j)<u_1(L_j)$ and $u_1(F_{j+1})>u_1(L_{j+1})$. By our assumption, $u_2(F_j)<u_2(L_j)$ and $u_2(F_{j+1})>u_2(L_{j+1})$ also hold. 

\smallskip

\noindent \textbf{Subcase 1.1} $u_1(F_j)\leq u_1(L_{j+1})$: set $\pi(1)=L_{j+1}$ and $\pi(2)=F_{j+1}$.

\smallskip

\noindent \textbf{Subcase 1.2} $u_1(F_j)>u_1(L_{j+1})$: set $\pi(1)=F_j$ and $\pi(2)=L_j$

\smallskip

In both subcase~1.1 and subcase~1.2, agent $1$ envies agent~$2$'s bundle, but this envy can be eliminated by deleting $j+1$ from agent~$2$'s bundle. This is because $u_1(L_{j+1}) \geq u_1(F_j)$ in subcase~1.1 and $u_1(F_j) > u_1(L_{j+1})$ in subcase~1.2.

\medskip

\noindent \textbf{Case 2.} $u_1(S)<0$ where $S$ is the entire set of $m$ items.

As $u_1(F_0)=u_1(L_m)=u_1(\emptyset)=0$ and $u_1(L_0)=u_1(F_m)=u_1(S)<0$, there exists an index $j \in \{0,\ldots,m-1\}$ such that $u_1(F_j)>u_1(L_j)$ and $u_1(F_{j+1})<u_1(L_{j+1})$. By our assumption, $u_2(F_j)>u_2(L_j)$ and $u_2(F_{j+1})<u_2(L_{j+1})$ also hold. 

\smallskip

\noindent \textbf{Subcase 2.1} $u_1(F_j)\geq u_1(L_{j+1})$: set $\pi(1)=F_{j+1}$ and $\pi(2)=L_{j+1}$.

\smallskip

\noindent \textbf{Subcase 2.2} $u_1(F_j)<u_1(L_{j+1})$: set $\pi(1)=L_j$ and $\pi(2)=F_j$.

\smallskip

In both subcase~2.1 and subcase~2.2, agent $1$ envies agent~$2$'s bundle, but this envy can be eliminated by deleting $j+1$ from agent~$1$'s bundle. This is because $u_1(F_j) \ge u_1(L_{j+1})$ in subcase~2.1 and $u_1(L_{j+1}) > u_1(F_j)$ in subcase~2.2.

\medskip

This concludes the proof of the theorem.
\end{proof}

\subsection{An Interesting Example}
\label{sec:counter-example}

The difference between EF1 and EFX allocations is well illustrated by the fact that, in contrast to the EF1 case, not even the existence of an \gse\ allocation is ensured.

\thmnmnai*
\begin{proof}
Let $S=\{1,2,3\}$ be a set of three items and let $u:2^S\rightarrow\mathbb{R}$ be defined by $u(\emptyset)=0$, $u(1)=1$, $u(2)=2$, $u(3)=0$, $u(12)=3$, $u(13)=0$, $u(23)=3$, and $u(123)=4$ (see Figure~\ref{fig:nm_na_i}).  We claim that there is no \gse\ allocation for two agents.

\begin{figure}[t]
    \centering
    \includegraphics[width=0.3\textwidth]{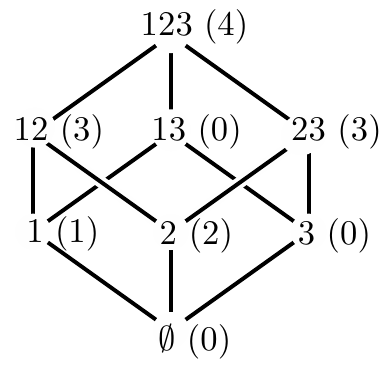}
    \caption{Illustration of Theorem~\ref{thm:nm_na_i}. The lattice of subsets of $S=\{1,2,3\}$, together with the corresponding utility value (in parenthesis).}
    \label{fig:nm_na_i}
\end{figure}

As the utility functions are identical, we may assume that the first agent receives more items. 
If all the items are allocated to agent $1$, then $u_1(\pi(1))=u(123)=4$ and $u_2(\pi(2))=u(\emptyset)=0$, hence agent $2$ envies agent $1$. However, this is a non-EFX envy as $1,3\in S^+(123)$, but deleting either item $1$ or $3$ from $123$ results in a set with positive utility value.

If $\pi(1)=12$ and $\pi(2)=3$, then $u_1(\pi(1))=u(12)=3$ and $u_2(\pi(2))=u(3)=0$, hence agent $2$ envies agent $1$. This is a non-EFX envy as $1\in S^+(12)$, but $2=u(2)>u(3)=0$. If $\pi(1)=23$ and $\pi(2)=1$, then $u_1(\pi(1))=u(23)=3$ and $u_2(\pi(2))=u(1)=1$, hence agent $2$ envies agent $1$. This is a non-EFX envy as $3\in S^+(23)$, but $2=u(2)>u(1)=1$. Finally, if $\pi(1)=13$ and $\pi(2)=2$, then $u_1(\pi(1))=u(13)=0$ and $u_2(\pi(2))=u(2)=2$, hence agent $1$ envies agent $2$. This is a non-EFX envy as $3\in S^-(13)$, but $1=u(1)<u(2)=2$. 
\end{proof}

\begin{rem}
Although the utility function defined in the proof of Theorem~\ref{thm:nm_na_i} is non-monotone, it is as close to being monotone as possible in the sense that it can be made monotone by increasing its value on a single set (namely $13$) by one. Note that an EF1 allocation exists: just set $\pi(1)=13$ and $\pi(2)=2$, then the envy of agent $1$ can be eliminated by deleting $2$ from $\pi(2)$.
\end{rem}

\section{\texorpdfstring{\gsee}{EFX+0} Allocations of Chores with Identical Utilities}
\label{sec:chores-efx}

Envy-freeness and its relaxations have been extensively studied in the setting of indivisible goods. A closely related problem is the fair division of {\em chores}. Here we assume every utility function $u_i$ is monotone non-increasing, i.e., $X \subseteq Y$ implies $u_i(X) \ge u_i(Y)$.

Our goal is to find an \gsee\ allocation $\pi$ of $S$. Note that for identical monotone non-increasing utilities this requires that 
\begin{description}
\item[(\namedlabel{star}{$\star$})]~ for any pair of agents $i,j$, $u_i(\pi(i) - s) \ge u_i(\pi(j))$ for all $s \in \pi(i)$.
\end{description}
We now show that such an allocation of chores always exists when (i)~$n = 2$ or (ii)~all $n$ agents have identical utilities. First we show the following result.

\thmchores*
\begin{proof}
Our algorithm is presented as Algorithm~\ref{alg:chores}. The joint utility function of the agents is denoted by $u$. The algorithm runs in rounds, and we maintain a pool $P$ of unallocated chores and an allocation that is \gsee\ on the chores currently assigned to the agents. Initially, we set $P = S$. Interestingly, chores once allocated to some agent may go back to the pool $P$ in a later round. 

Consider any round and let $\pi = \langle \pi(1),\ldots,\pi(n)\rangle$ be the \gsee\ allocation at the start of this round and let $\pi' = \langle \pi'(1),\ldots,\pi'(n)\rangle$ be the \gsee\ allocation at the end of this round. Let $P$ (resp., $P'$) be the set of unallocated chores at the start (resp., end) of this round. We will ensure that at least one of the following two conditions is satisfied: 
\begin{itemize}
\item[(i)] $|P'| < |P|$ and $\sum_{i\in N} u(\pi'(i)) \le \sum_{i\in N} u(\pi(i)))$, i.e., the number of allocated chores increases in this round and {\em utilitarian welfare} does not increase,
\item[(ii)] $\sum_{i\in N} u(\pi'(i)) < \sum_{i\in N} u(\pi(i)))$, i.e., utilitarian welfare strictly decreases in this round.
\end{itemize}
In other words, in each round of our algorithm, we either decrease utilitarian welfare or we increase the number of allocated chores without increasing utilitarian welfare. We now describe a single round in detail. 

Since all agents have the same utilities, the agents can be ordered in terms of the utilities of their bundles. So at any point in time, there is at least one agent who can be called a {\em happiest} agent: one who does not envy any other agent. In each round, we pick any unallocated chore $s$ and add $s$ to the bundle of a happiest agent (call this agent $k$). There are two cases:
\medskip

\noindent\textbf{Case 1.} 
Suppose $\langle \pi(1),\ldots,\pi(k)+s,\ldots,\pi(n)\rangle$ is \gsee. Then set $\pi'(k) = \pi(k) + s$ and $\pi'(i) = \pi(i)$ for $i \ne k$. Observe that the allocation $\pi' = \langle \pi'(1),\ldots,\pi'(n)\rangle$ satisfies condition~(i) given above.

\medskip

\noindent\textbf{Case 2.}
Suppose $\langle \pi(1),\ldots,\pi(k)+s,\ldots,\pi(n)\rangle$ is not \gsee. By observation \eqref{star}, $u(\pi(k) + s - s') < u(\pi({\ell}))$ for some $s' \in \pi(k)$ and $\ell \in N$. Thus we can find an inclusionwise minimal subset $\pi'(k) \subset \pi(k) + s$ such that $\pi'(k)$ is an envious bundle, i.e., $u(\pi'(k)) < u(\pi(\ell))$ for some $\ell \in N$. Since $\pi'(k)$ is a minimal envious subset of $\pi(k) + s$, for every $X \subset \pi'(k)$, we have $u(X) \ge u(\pi(i))$ for all $i \in N$. 

Now the chores in $(\pi(k) + s)\setminus \pi'(k)$ are thrown back into the pool $P$, i.e., 
$P' = P \cup (\pi(k) + s)\setminus \pi'(k)$. Let $\pi'(i) = \pi(i)$ for all $i \ne k$.

\begin{cl}
The allocation $\pi' = \langle \pi'(1),\ldots,\pi'(n)\rangle$ is an \gsee\ allocation. 
\end{cl}
\begin{proof}
We need to show that $u(\pi'(i) - r) \ge u(\pi'(j))$ for all $i, j \in N$ and $r \in \pi'(i)$. Since $\pi'(t) = \pi(t)$ for all $t \ne k$ and $\pi$ is \gsee, for all $i,j$ such that neither $i$ nor $j$ is $k$, this holds.

{\em When $i = k$:} Since $\pi'(k)$ is a minimal envious subset of $\pi(k) + s$, we have $u(\pi'(k) - r) \ge u(\pi(j)) = u(\pi'(j))$ for all $r \in \pi'(k)$ and $j \ne k$.

{\em When $j = k$:} We claim that $u(\pi'(k)) < u(\pi(k))$. Indeed, $u(\pi'(k)) < u(\pi(\ell))$ for some $\ell\in[n]$ since $\pi'(k)$ is envious, and $u(\pi(\ell)) \le u(\pi(k))$ for all $\ell\in[n]$ since $k$ is a happiest agent in $\pi$. Thus for any $i \ne k$, we have $u(\pi'(i) - r) = u(\pi(i) - r) \ge u(\pi(k)) > u(\pi'(k))$ for all $r \in \pi'(i)$, where the inequality $u(\pi(i) - r) \ge u(\pi(k))$ follows from $\pi$ being \gsee.  
\end{proof}

Since $u(\pi'(k)) < u(\pi(k))$ and $\pi'(t) = \pi(t)$ for all $t \ne k$, $\sum_{i\in[n]} u(\pi'(i)) < \sum_{i\in[n]} u(\pi(i)))$. Thus the allocation $\pi' = \langle \pi'(1),\ldots,\pi'(n)\rangle$ satisfies condition~(ii) given above.

\medskip

Note that there can be at most $|S|$ consecutive rounds where condition~(i) holds. After that either all chores are allocated or there is a round where condition~(ii) holds, i.e., utilitarian welfare strictly decreases. Thus the number of rounds is at most $|S|/{\Delta}$ where $\Delta = \min \{|u(X) - u(Y)|: X,Y \subseteq S \text{ and } u(X) \ne u(Y)\}$ is the minimum difference between distinct utilities. Therefore our algorithm always terminates, concluding the proof of the theorem.
\end{proof}

\begin{algorithm}[t!]
  \caption{Finding an \gsee\ allocation of chores for monotone, identical utility functions.}\label{alg:chores}
  \begin{algorithmic}[1]
    \Statex \textbf{Input:} Set $N$ of agents, set $S$ of items, monoton utility function $u:2^S\rightarrow\mathbb{R}_-$.
    \Statex \textbf{Output:} An \gsee\ allocation. 
    \State Let $P:=S$.
    \State Set $\pi(i):=\emptyset$ for $i\in N$.
    \While{$P$ is not empty} \label{st:while_2}
        \State Pick an unallocated chore $s\in P$.
        \State Pick a happiest agent $k$.
        \If{$(\pi(1),\dots,\pi(k)+s,\dots,\pi(n))$ is \gsee}
            \State $\pi(k)\leftarrow \pi(k)+s$
            \State $P\leftarrow P-s$
        \Else
            \State \begin{varwidth}[t]{0.9\linewidth} Let $X\subset\pi(k)+s$ be inclusionwise minimal s.t. $u(X)<u(\pi(l))$ for some $l\in N$. \end{varwidth}
            \State $\pi(k)\leftarrow X$
            \State $P\leftarrow P\cup(\pi(k)+s)\setminus X$
        \EndIf
    \EndWhile
    \State \textbf{return} $\pi$.
  \end{algorithmic}
\end{algorithm}

The cut-and-choose protocol of Plaut and Roughgarden~\cite{plaut2018almost} implies the following corollary.

\corchorestwoagents*
\begin{proof}
Let $u_1$ and $u_2$ be the utility functions of agent~$1$ and agent~$2$, respetively. By Theorem~\ref{thm:chores-identical-val}, there exists an \gsee\ allocation when the utility function of both agents is $u_1$. Take such an allocation; this defines a partition of $S$ into two parts $S_1$ and $S_2$. Now let the second agent choose among $S_1$ and $S_2$ based on her preferences. Clearly, agent~$2$ will have no envy as she chooses the set with better utility with respect to $u_2$. On the other hand, even if agent~$1$ envies agent~$2$, this is an \gsee\ envy due to the construction of the allocation. 
\end{proof}

\section{Boolean and Negative Boolean Utility Functions}
\label{sec:two-utility-functions}

In this section we consider special cases of non-identical utility functions: Boolean and negative Boolean utilities. 

\subsection{Boolean Utilities}
\label{sec:boolean}

Recall that a utility function $u$ is called Boolean if $u(X)\in\{0,1\}$ for $X\subseteq S$. 

\thmboolean*
\begin{proof}
 We construct an \gese\ solution by assigning bundles to agents one-by-one using Algorithm~\ref{alg:boolean}. After each step, we will refer to unassigned items and to agents without bundles as \emph{remaining} items and agents, respectively.

At a general step of the algorithm, we take an inclusionwise minimal subset $X$ of the remaining items that has utility value $1$ for at least one of the remaining agents, and we assign $X$ to one of the agents $i$ with $u_i(X)=1$. If no such set exists, then we pick an arbitrary agent $i$ and set $X$ to be the empty set. Then we delete the members of $X$ from $S$ and $i$ from $N$. When only one agent remains, we assign the remaining set of items to her.

We claim that the solution thus obtained is an \gese\ allocation. We may assume that the order of the agents in which their bundles get fixed is $1,\dots,n$. Note that the sequence $u_1(\pi(1)),\dots,u_n(\pi(n))$ is monotone decreasing, that is, it consists of a sequence of $1$'s followed by a sequence of $0$'s (any of these two parts can be empty). Therefore if agent $i$ envies $j$, then $j<i$, $u_i(\pi(j))=1$, and $u_i(\pi(i))=0$. However, by the choice of $\pi(j)$ in Step~\ref{st:min} of the algorithm, $u_i(\pi(j)-s)=0$ for every $s\in\pi(j)$. Moreover, $S^+_i(\pi(j))\neq\emptyset$ and the statement follows. 
\end{proof}

\begin{algorithm}[t!]
  \caption{Finding an \gese\ allocation for Boolean utility functions.}\label{alg:boolean}
  \begin{algorithmic}[1]
    \Statex \textbf{Input:} Set $N$ of agents, set $S$ of items, utility functions $u_i:2^S\rightarrow\{0,+1\}$.
    \Statex \textbf{Output:} An \gese\ allocation. 
    \While{$|N|\geq 2$}
        \If{there exist $i\in N$ and $X\subseteq S$ s.t. $u_i(X)=1$}
            \State  \begin{varwidth}[t]{0.9\linewidth} Let $X\subseteq S$ be inclusionwise minimal s.t. $u_i(X)=1$ for some $i\in N$. Let $i$ be an agent with $u_i(X)=1$. Set $\pi(i):=X$. \label{st:min} \end{varwidth}
            \State $N\leftarrow N-i$
            \State $S\leftarrow S-X$
        \Else
            \State Let $i\in N$ arbitrary. Set $\pi(i):=\emptyset$.
            \State $N\leftarrow N-i$
        \EndIf
    \EndWhile
    \State For the only remaining agent $j$ in $N$, set $\pi(j):=S$.
    \State \textbf{return} $\pi$
  \end{algorithmic}
\end{algorithm}

Note that an \gsee\ allocation does not necessarily exist for Boolean allocations. To see this, consider the instance with three items and two agents with identical utilities defined as follows: $u(X)=0$ if $|X|\leq 1$ and $u(X)=1$ otherwise. It is not difficult to verify that no allocation satisfies the conditions of \ref{gsee}.

\begin{rem}
The proof of Theorem~\ref{thm:boolean} is non-algorithmic, as Step~\ref{st:min} of Algorithm~\ref{alg:boolean} asks for an inclusionwise minimal set with utility value $1$ for at least one of the remaining agents. Of course the complexity of finding such a set depends on how the utility functions are given, but the difficulty is well illustrated by the fact that the problem is equivalent to determining a satisfying assignment of a Boolean function using a minimum number of true variables. 

One of the most common representations of Boolean functions are \emph{conjunctive normal forms} (CNFs), the conjunctions of clauses which are elementary disjunctions of literals. A CNF is called \emph{pure Horn} if every clause in it contains exactly one positive literal, and a Boolean function is pure Horn if it admits a pure Horn CNF representation. Pure Horn functions form a fundamental subclass of Boolean functions admitting interesting structural and computational properties. Among others, SAT is solvable for this class in linear time \cite{DOWLING1984267}. 

If each utility function $u_i$ is represented by a pure Horn CNF, then a set $X$ satisfying the conditions of Step~\ref{st:min} can be determined in polynomial time with the help of the the so-called \emph{forward chaining procedure}. For further details, we refer the interested reader to \cite{hammer1993optimal}.
\end{rem}

\subsection{Negative Boolean Utilities}
\label{sec:negboolean}

Recall that a utility function $u$ is called negative Boolean if $-u$ is Boolean.

\thmnegboolean*
\begin{proof}
The algorithm is presented as Algorithm~\ref{alg:negboolean}. The joint utility function of the agents is denoted by $u$. We start with an arbitrary allocation, say, $\pi(1)=S$ and $\pi(i)=\emptyset$ otherwise. In a general step of the algorithm, we pick a pair $i,j\in N$ of agents such that there is a non-EFX envy from $i$ to $j$. As we are considering \gsee\ allocations and the utility function is negative Boolean, this means that there exists an item $s\in\pi(i)$ such that $-1=u(\pi(i))=u(\pi(i)-s)<u(\pi(j))=0$. The algorithm moves such an item from the bundle of agent $i$ to that of agent $j$.

It suffices to show that the algorithm terminates after a polynomial number of steps. Recall that the numbers of agents and items are denoted by $n$ and $m$, respectively. Define 
\begin{equation*}
    \varphi(\pi)=m\cdot\sum_{\substack{i\in N\\u(\pi(i))=-1}}1+\sum_{\substack{i\in N\\u(\pi(i))=0}}|\pi(i)|.
\end{equation*} 
In other words, $\varphi$ counts $m$ times the number of agents with utility value $-1$, plus the total number of items in bundles having utility value $0$.

\begin{cl} \label{cl:strictly}
$\varphi(\pi))$ strictly increases throughout the algorithm.
\end{cl}
\begin{proof}
Let $\pi'$ denote the allocation obtained from $\pi$ by moving $s$ from $\pi(i)$ to $\pi(j)$, that is, $\pi'(i)=\pi(i)-s$, $\pi'(j)=\pi(j)+s$, $\pi'(k)=\pi(k)$ otherwise. By $u(\pi(i))=-1$, we have $\pi(i)\neq\emptyset$ and so $|\pi(j)|\leq m-1$. 

If $u(\pi'(j))=-1$, then 
\begin{align*}
\varphi(\pi')
{}&{}=m\cdot\sum_{\substack{i\in N\\u(\pi'(i))=-1}}1+\sum_{\substack{i\in N\\u(\pi'(i))=0}}|\pi'(i)|\\
{}&{} \geq 
m\cdot\left(\left(\sum_{\substack{i\in N\\u(\pi(i))=-1}}1\right)+1\right)+\left(\sum_{\substack{i\in N\\u(\pi(i))=0}}|\pi(i)|-(m-1)\right)\\
{}&{}=
m\cdot\sum_{\substack{i\in N\\u(\pi(i))=-1}}1+\sum_{\substack{i\in N\\u(\pi(i))=0}}|\pi(i)|+1\\
{}&{}=
\varphi(\pi)+1,
\end{align*}
while if $u(\pi'(j))=0$, then 
\begin{align*}
\varphi(\pi')
{}&{}=m\cdot\sum_{\substack{i\in N\\u(\pi'(i))=-1}}1+\sum_{\substack{i\in N\\u(\pi'(i))=0}}|\pi'(i)|\\
{}&{}=
m\cdot\sum_{\substack{i\in N\\u(\pi(i))=-1}}1+\sum_{\substack{i\in N\\u(\pi(i))=0}}|\pi(i)|+1\\
{}&{}=
\varphi(\pi)+1.
\end{align*}
This concludes the proof of the claim.
\end{proof}

As $\varphi$ takes integer values upper bounded by $m\cdot n$, Claim~\ref{cl:strictly} implies that the algorithm terminates after a polynomial number of steps, and the theorem follows.
\end{proof}

\begin{algorithm}[t!]
  \caption{Finding an \gsee\ allocation for identical negative Boolean utility functions.}\label{alg:negboolean}
  \begin{algorithmic}[1]
    \Statex \textbf{Input:} Set $N$ of agents, set $S$ of items, utility function $u:2^S\rightarrow\{0,-1\}$.
    \Statex \textbf{Output:} An \gsee\ allocation. 
    \State Set $\pi(1):=S$ and $\pi(i)=\emptyset$ otherwise.
    \While{$\pi$ is not \gsee}
        \State Let $i,j\in N$ s.t. $u(\pi(i))<u(\pi(j))$ and the envy is non-EFX.
        \State Choose $s\in\pi(i)$ such that $u(\pi(i)-s)<u(\pi(j))$.
        \State $\pi(i)\leftarrow\pi(i)-s$
        \State $\pi(j)\leftarrow\pi(j)+s$
    \EndWhile
    \State \textbf{return} $\pi$
  \end{algorithmic}
\end{algorithm}

Note that an \gese\ allocation does not necessarily exist for negative Boolean allocations. To see this, consider the instance with three items and two agents with identical utilities defined as follows: $u(X)=0$ if $|X|\leq 1$ and $u(X)=-1$ otherwise. It is not difficult to verify that no allocation satisfies the conditions of \ref{gese}.

\begin{rem}
A natural adaptation of Algorithm~\ref{alg:boolean} to the negative Boolean setting would be to always choose an inclusionwise minimal subset of the remaining items in Step \ref{st:min} that has $-1$ utility value. However, such an approach has no control over the set that is allocated to the last agent, and so a non-EFX envy might be present.
\end{rem}

\section{Conclusions}
\label{sec:conc}

The present paper focused on the concept of envy-freeness and its relaxations, envy-freeness up to one item and envy-freeness up to any item. Concerning EF1 allocations, we presented a polynomial-time algorithm for finding one for two agents with arbitrary utility functions. We extended the notion of EFX allocations to non-monotone, non-additive utilities, and settled the existence or non-existence of such solutions in various settings. We showed that an \gsee\ allocation of chores always exists for monotone identical utility functions.
For the classes of Boolean and identical negative Boolean utilities, we verified the existence of \gese\ and \gsee\ allocations, respectively.

Tables~\ref{tab:ef1res} and \ref{tab:efxres} show that the existence or non-existence of a fair solution is still open in many cases. Among them, we would like to draw attention to two open problems that seem to be particularly interesting.

\begin{qu}
Does there always exist an EF1 allocation for non-monotone, non-additive, identical utility functions?
\end{qu}

\begin{qu}
Does there always exist an \gse\ allocation for monotone, additive, non-identical utility functions?
\end{qu}

\section*{Acknowledgements}
 
Krist\'of B\'erczi was supported by the J\'anos Bolyai Research Fellowship of the Hungarian Academy of Sciences and by the ÚNKP-19-4 New National Excellence Program of the Ministry for Innovation and Technology. 
Naoyuki Kamiyama was supported by JST, PRESTO Grant Number JPMJPR1753, Japan.
Projects no. NKFI-128673 and no. ED\_18-1-2019-0030 (Application-specific highly reliable IT solutions) have been implemented with the support provided from the National Research, Development and Innovation Fund of Hungary, financed under the FK\_18 and the Thematic Excellence Programme funding schemes, respectively. 
The work was supported by the Research Institute for Mathematical Sciences, an International Joint Usage/Research Center located in Kyoto University.

\bibliographystyle{abbrv}
\bibliography{efxbib}

\end{document}